\documentclass[twocolumn]{ijns}
\usepackage{graphicx,epsfig,color}
\usepackage[normal]{subfigure}
\usepackage{amsmath,amsthm,amssymb}
\usepackage{floatflt}  
\usepackage{setspace}
\usepackage{rotating,algorithm,algorithmic}
\usepackage{enumerate}
\usepackage{tabularx}

\def\markboth#1#2{\def\leftmark{#1}\def\rightmark{#2}}

\newtheorem{definition}{Definition}
\newtheorem{theorem}{Theorem}
\newtheorem{lemma}{Lemma}

\title{\huge\bf Double Circulant Self-Dual Codes From Generalized Cyclotomic Classes of Order Two}
\author{Wenpeng Gao$^1$ and Tongjiang Yan$^2$\\
\medskip {\small\it(Corresponding author: Tongjiang Yan)}~\\
{\normalsize College of Sciences, China University of Petroleum$^1$}\\
{\normalsize Qingdao, 266555, China}\\
{\normalsize College of Sciences, China University of Petroleum; Shandong Provincial Key Laboratory of Computer Networks $^2$}\\
{\normalsize Key Laboratory of Applied Mathematics(Putian University), Fujian Province University $^2$}\\
{\normalsize Qingdao, 266555, China; Jinan, 250014, China; Fujian Putian, 351100, China }\\
{\normalsize (Email: yantoji@163.com)}\\
{\normalsize\it (Received 2019; revised and accepted 2019) }
}

\date
\newpage
\begin{document}
\maketitle
\thispagestyle{headings}
\begin{abstract}
In this paper, constructions of some double circulant self-dual codes by generalized cyclotomic classes of order two are presented. This technique is applied to [72, 36, 12] binary highest know self-dual codes to obtain self-dual codes over GF(2) and [32, 16, 8] almost optimal self-dual codes over GF(4). Based on the properties of generalized cyclotomy, some of these codes can be proved to possess good minimum weights.
\vspace*{0.1cm}
~
\end{abstract}

\section{Introduction}
Self-dual codes have important applications in transmission~\cite{Garcia2010Application}. Constructing self-dual codes which have good minimum weight has been an important research problem~\cite{Alahmadi2017On}. And a number of papers devoted to constructing self-dual codes ~\cite{Alahmadi2018On}~\cite{Arasu2001Self}. Most of self-dual codes which are double circulant
codes have high minimum distance ~\cite{Gaborit2002Quadratic}. The generalized cyclotomic classes have wide applications in constructing sequences [9,13,20], cyclic codes ~\cite{Ding2012Cyclic}~\cite{Ding2013Cyclic} and difference sets~\cite{Feng2012Cyclotomic}. There are a number of results about generalized cyclotomic classes of order two~\cite{Ding1998Autocorrelation}~\cite{Wang2016Generalized}. Motivated by the constructions of double circulant codes by cyclotomic classes of order four~\cite{Tao2015Fourth}, we give constructions of double circulant codes by generalized cyclotomic classes of order two. All computations have been done by MAGMA V2.20-4~\cite{Bosma1997The} on a 3.40 GHz CPU.

This paper is organized as follows. In Section II, we present definitions and some preliminaries about self-dual codes and generalized cyclotomic classes. In Section III, three infinite families of  self-dual codes are given. Section IV concludes the paper.

\section{Preliminaries}
\subsection{Self-Dual codes}
A linear code $C$ of length $n$ and dimension $k$ over finite field GF$(l)$ is a $k$-dimensional subspace of GF$(l)^{n}$, where $l$ is a prime power. A generator matrix $G$ of the code $C$ is a $k \times n$ matrix over GF$(l)$. The elements of $C$ are called codewords. The Euclidean inner product is defined by
\begin{center}
$(x, y)=\sum_{i=1}^{n}x_{i}y_{i},$
\end{center}
where $x = (x_1, \dots, x_n)$ and $y = (y_1, \dots, y_n)$.  The  Euclidean dual code, denoted by $C^\top$, of linear code $C$ is defined by
\begin{center}
$C^\top=\{x\in$ GF$(l)^n|(x,c)=0,  \textrm{for all } c \in C\}.$
\end{center}
And we have $dim(C)+dim(C^\top) = n$. $C$ is said to be Euclidean self-orthogonal if $C \subset C^\top$ and Euclidean self-dual if $C = C^\top$. If $C$ is dual-self code, the dimension of $C$ is $n/2$. From now on, what we mean by self-dual is Euclidean self-dual.

\subsection{Hamming Distance}
The Hamming distance of codewords $x = (x_1, \dots, x_n)$ and $y = (y_1, \dots, y_n)$, denoted by $d(x, y)$, is defined to be the number of places at which $x$ and $y$ differ. The Hamming weight $wt(x)$ of a codeword $x = (x_1, \dots, x_n)$ is $wt(x) = d(x, 0)$ and the minimum Hamming distance $d(C)$ of a code $C$ is equal to the minimum nonzero Hamming weight of all codewords in $C$. Then for the self-dual codes, we have the following result.
\begin{theorem}\label{Theorem 1}
~\cite{Nebe2006Self}~\cite{Rains1998Shadow} Let $C$ be a self-dual code over \rm{GF}$(l)$ of length $n$ and minimum Hamming distance $d(C)$. Then we have
\begin{enumerate}[(i)]
\item If $l=2$, then
\begin{center}
$d(C)\leq
\begin{cases}
4\left\lfloor \frac{n}{24} \right\rfloor+4, & \textrm{if $n\not\equiv 22 \pmod{24}$;}\\
4\left\lfloor \frac{n}{24} \right\rfloor+6, & \textrm{if $n\equiv 22 \pmod{24}$.}\\
\end{cases}$
\end{center}
\item If $l=3$, then $d(C)\leq 3\left\lfloor \frac{n}{12} \right\rfloor+3$.
\item If $l=4$, then $d(C)\leq 4\left\lfloor \frac{n}{12} \right\rfloor+4$.
\item $d(C)\leq \left\lfloor \frac{n}{2} \right\rfloor+1$ for $l\neq 2,3,4$.
\end{enumerate}
\end{theorem}
The code $C$ is called extremal if the above equality holds. A self-dual code is called optimal if it has the highest possible minimum distance for its length.
\subsection{Generalized Cyclotomy and Generalized Cyclotomic Number}

Let $n = pq$, where $p$ and $q$ are distinct odd primes with $\gcd(p - 1, q - 1) = d$. The Chinese Remainder Theorem \cite{Ding1996Chinese} guarantees that there exists common primitive roots of both $p$ and $q$. Let $g$ be a fixed common primitive root of both $p$ and $q$. Let $e := {\rm ord}_n(g)$ denote the multiplicative order of $g$ modulo $n$, then
\begin{equation*}
\begin{aligned}
{\rm ord}_n(g) &= {\rm lcm}({\rm ord}_p(g),{\rm ord}_q(g))\\
 &= {\rm lcm}(p-1, q-1) = \frac{(p-1)(q-1)}{d}.
\end{aligned}
\end{equation*}
Let $x$ be an integer satisfying
\begin{center}
$x\equiv g \pmod{p},\quad x\equiv 1 \pmod{q}.$
\end{center}
Whiteman \cite{Stanton1958A} defined the generalized cyclotomic classes
\begin{equation*}\label{generalized cyclotomic classes}
C_i=\{g^sx^i : s=0, 1, \dots , e-1\}, i=0, 1, \dots, d-1.
\end{equation*}
In this paper, we shall always assume that $p, q$ are distinct primes with $\gcd(p-1, q-1) = 2$.
\begin{equation*}
\begin{split}
P&=\{p, 2p, \dots, (q-1)p\},\\
Q&=\{q, 2q, \dots, (p-1)q\},\\
R&=\{0\}.
\end{split}
\end{equation*}
The generalized cyclotomic numbers corresponding of order two are defined by
\begin{equation*}
(i,j) =\left|(C_i+1) \cap C_j\right|,   \textrm{for all $i, j =0, 1$.}
\end{equation*}
\begin{lemma}\label{Lemma 5}
 \cite{Ding1998Autocorrelation}$\quad -1 \in C_0 $  \textrm{if $(p-1)(q-1)/4$ is even}; $-1 \in C_1$  \textrm{if $(p-1)(q-1)/4$ is odd}.\\
\end{lemma}
\begin{lemma}\label{Lemma 6}
 \cite{Stanton1958A} If $(p-1)(q-1)/4$ is even,
 \begin{equation*}
\begin{aligned}
 (0,0)&=(1,0)=(1,1)=\frac{(p-2)(q-2)+1}{4},\\
 (0,1)&=\frac{(p-2)(q-2)-3}{4}.
\end{aligned}
\end{equation*}
 \indent If $(p-1)(q-1)/4$ is odd,
 \begin{equation*}
\begin{aligned}
(0,0)&=\frac{(p-2)(q-2)+3}{4},\\
(0,1)&=(1,0)=(1,1)=\frac{(p-2)(q-2)-1}{4}.
\end{aligned}
\end{equation*}
\end{lemma}

\subsection{General Results }

Let $m_0$, $m_1$, $m_2$, $m_3$ and $m_4$ be elements of GF$(l)$. The matrix $C_n(m_0, m_1$, $m_2, m_3$, $m_4)$ is the $n\times n$ matrix on GF$(l)$  with components $c_{ij}$ , $1 \leq i, j \leq n$ ,
$$c_{ij}=
 \begin{cases}
m_0  & \textrm{if $j=i$,}\\
m_1  & \textrm{if $j-i \in P$,}\\
m_2  & \textrm{if $j-i \in Q$,}\\
m_3  & \textrm{if $j-i \in C_0$,}\\
m_4  & \textrm{if $j-i \in C_1$,}\\
\end{cases}$$
Define by $I_n$ and $J_n$ the identity and the all-one square $n \times n$ matrices. Then $C_n(1, 0, 0, 0, 0) = I_n $ and $C_n(1, 1, 1, 1, 1) = J_n$. Denote $P_n:= C_n(0, 1, 0, 0, 0),$ $Q_n:= C_n(0, 0, 1, 0, 0),$ $A_1:= C_n(0, 0, 0, 1, 0)$ and  $A_2:= C_n(0, 0,$ $ 0, 0, 1)$.

\begin{lemma}\label{Lemma 2.10}
If $p$ is a prime of the form $4\omega+1$ and $q$ is a prime of the form $4\omega^{'}+3$. Then
\begin{equation*}
\begin{aligned}
A_1=&A_2^\top, A_2=A_1^\top, P_n=P_n^\top, Q_n=Q_n^\top,\\
P_nA_1=&A_1P_n=Q_n+A_2, P_nA_2=A_2P_n=Q_n+A_1,\\
Q_nA_1=&A_1Q_n=A_1,  Q_nA_2=A_2Q_n=A_2,\\
P_nQ_n=&Q_nP_n=A_1+A_2,P_n^2=P_n, Q_n^2=Q_n,\\
A_1^2=&(1,0)A_1+(0,1)A_2,\\
 A_2^2=&(0,1)A_1+(1,0)A_2,\\
A_1A_2=&Q_n+(0,0)A_1+(1,1)A_2,\\
A_2A_1=&Q_n+(1,1)A_1+(0,0)A_2.
\end{aligned}
\end{equation*}
\end{lemma}
\begin{proof}
The proof is straightforward from the definitions of these sets.
\end{proof}

\begin{lemma}\label{Lemma 2.9}
If $p$ is a prime of the form $4\omega+3$ and $q$ is a prime of the form $4\omega^{'}+1$. Then
\begin{equation*}
\begin{aligned}
A_1=&A_2^\top, A_2=A_1^\top, P_n=P_n^\top ,Q_n=Q_n^\top,\\
P_nA_1=&A_1P_n=A_1,\quad P_nA_2=A_2P_n=A_2,\\
Q_nA_1=&A_1Q_n=P_n+A_2,\quad Q_nA_2=A_2Q_n=P_n+A_1,\\
P_nQ_n=&Q_nP_n=A_1+A_2,P_n^2=P_n, Q_n^2=Q_n,\\
A_1^2=&(1,0)A_1+(0,1)A_2,\\
 A_2^2=&(0,1)A_1+(1,0)A_2,\\
A_1A_2=&P_n+(0,0)A_1+(1,1)A_2,\\
A_2A_1=&P_n+(1,1)A_1+(0,0)A_2.
\end{aligned}
\end{equation*}
\end{lemma}
\begin{proof}
The proof is straightforward from the definitions of these sets.
\end{proof}

\indent The following result will be needed in the sequel.
\begin{equation*}
\begin{aligned}
&(m_0I_n + m_1P_n + m_2Q_n + m_3A_1 + m_4A_2)\\
&(m_0I_n+ m_1P_n + m_2Q_n + m_3A_1 + m_4A_2)^{\top}\\
=&a_0I_n + a_1P_n + a_2Q_n + a_3A_1 + a_4A_2.
\end{aligned}
\end{equation*}
\begin{lemma}\label{Lemma 2.11}
If $p=4\omega+1$ and $q$ =$4\omega^{'}+3$. Then
\begin{equation*}
\begin{aligned}
a_0&=m_0^{2},\\
a_1&=m_1^{2}+2m_0m_1,\\
a_2&=m_2^{2}+m_3^{2}+m_4^{2}+2m_0m_2+2m_1(m_3+m_4),\\
a_3&=a_4\\
&=(m_0+m_1+m_2)(m_3+m_4)+2m_1m_2\\
&\quad +(0,0)(m_3^{2}+m_4^{2})+((1,0)+(0,1))m_3m_4.
\end{aligned}
\end{equation*}

If $p=4\omega+3$ and $q$ =$4\omega^{'}+1$. Then
\begin{equation*}
\begin{aligned}
a_0&=m_0^{2},\\
a_1&=m_1^{2}+m_3^{2}+m_4^{2}+2m_0m_1+2m_2(m_3+m_4),\\
a_2&=m_2^{2}+2m_0m_2, \\
a_3&=a_4\\
&=(m_0+m_1+m_2)(m_3+m_4)+2m_1m_2\\
&\quad +((1,0)+(0,1))m_3m_4+(0,0)(m_3^{2}+m_4^{2}).
\end{aligned}
\end{equation*}
\end{lemma}
For convenience, we denote
\begin{equation*}\label{m}
\begin{aligned}
&\overrightarrow{m}:=(m_0, m_1, m_2, m_3, m_4)\in GF(l)^5,\\
&D_0(\overrightarrow m):=a_0,  D_1(\overrightarrow m):=a_1,\\
&D_2(\overrightarrow m):=a_2,  D_3(\overrightarrow m):=a_3.
\end{aligned}
\end{equation*}
\begin{definition}\label{definition 1}
\cite{Macwilliams1977The}Let $GP_n(R)$ and $GB_n(\alpha, R)$ be codes with generator matrices of the form
$$\left(I_n,R\right)$$
$$
	\left(
	\begin{array}{ccccc}
	&& \alpha &1& \cdots 1\\
	&& -1 &&  \\
	&I_n& \vdots  && R  \\
	&& -1 &&  \\
	\end{array}
	\right),
$$
where $\alpha \in$ \rm{GF}$(l)$ and $R$ is an $n\times n$ circulant matrix. The codes $GP_n(R)$ and $ GB_n(\alpha, R)$ are called \textbf{pure double circulant codes} and \textbf{bordered double circulant codes}, respectively.
\end{definition}
For convenience, Let $GP_n(\overrightarrow m)$ be generator matrix of $GP_n(R)$ and $GB_n(\alpha ,\overrightarrow m)$ be generator matrix of $GB_n(\alpha, R)$.
\begin{equation*}
\begin{aligned}
&GP_n(\overrightarrow m)\\
&\quad :=GP_n(m_0I_n + m_1P_n + m_2Q_n + m_3A_1 + m_4A_2),\\
&GB_n(\alpha ,\overrightarrow m)\\
&\quad :=GB_n(\alpha ,m_0I_n + m_1P_n + m_2Q_n + m_3A_1 + m_4A_2).
\end{aligned}
\end{equation*}
\indent The main theorem of this section is:
\begin{theorem}\label{Theorem 2:}
 Let $\alpha \in$ \rm{GF}$(l)$ and $\overrightarrow m \in$ \rm{GF}$(l)^5$. Then
\begin{enumerate}[1)]
\item the code with generator matrix $GP_n(\overrightarrow m)$ is self-dual over \rm{GF}$(l)$ if and only if the following holds:
   \begin{enumerate}[a)]
       \item $D_0(\overrightarrow m)=-1,$
        \item $D_1(\overrightarrow m)=0,$
       \item $D_2(\overrightarrow m)=0,$
       \item $D_3(\overrightarrow m)=0,$
       \item $D_4(\overrightarrow m)=0.$
   \end{enumerate}
\item the code with generator matrix $GB_n(\alpha, \overrightarrow m)$ is self-dual over \rm{GF}$(l)$ if and only if the following holds:
    \begin{enumerate}[a)]
       \item $\alpha +n=-1,$
       \item $-\alpha+ m_0 +(p-1)m_1+(q-1)m_2\\ \qquad +\frac{(p-1)(q-1)}{2}(m_3 +m_4)=0,$
       \item $D_0(\overrightarrow m)=-2,$
       \item $D_1(\overrightarrow m)=-1,$
       \item $D_2(\overrightarrow m)=-1,$
       \item $D_3(\overrightarrow m)=-1.$
    \end{enumerate}
\end{enumerate}
\end{theorem}
\begin{proof}
The result follows from
\begin{equation*}
\begin{aligned}
&GP_n(\overrightarrow m)GP_n(\overrightarrow m)^{\top}\\
=&I_n+D_0I_n + D_1P_n + D_2Q_n + D_3A_1 + D_4A_2,\\
\textrm{and}\\
&GB_n(\alpha, \overrightarrow m)GB_n(\alpha, \overrightarrow m)^{\top}\\
=&I_{n+1}+
	\left(
	\begin{array}{cccc}
	 \alpha+n   &S  \cdots S\\
	 S &  \\
	 \vdots  & X  \\
	 S &  \\
	\end{array}
	\right),
\end{aligned}
\end{equation*}
where $X=J_n+D_0(\overrightarrow m)I_n+D_1(\overrightarrow m)P_n+D_2(\overrightarrow m)Q_n+D_3(\overrightarrow m)A_1 + D_4(\overrightarrow m)A_2$ and $S =-\alpha+ m_0 +(p-1)m_1+(q-1)m_2+ \frac{(p-1)(q-1)}{2}(m_3 +m_4)$.
\end{proof}

\section{Double circulant self-dual codes over fields}\label{section 3}

In this section, we construct two infinite families of self-dual codes over GF(2) and GF(4). As a preparation, we have the following two lemmas.
\begin{lemma}\label{lemma3.1}
Let $p=4\omega +3$ and $q=4\omega^{'}+1$ or $p=4\omega +1$ and $q=4\omega^{'}+3$. Then
\begin{enumerate}[1)]
\item If $\omega+\omega^{'}$ is even, which means $\frac{p+q}{4}$ is odd, \\
$(0,0)\equiv (1,0)\equiv (1,1)\equiv 0\pmod{2},  (0,1) \equiv 1\pmod{2}$.
\item If $\omega+\omega^{'}$ is odd, which means $\frac{p+q}{4}$ is even,  \\
$(0,0)\equiv (1,0)\equiv (1,1)\equiv 1\pmod{2}, (0,1) \equiv 0\pmod{2}$.
\end{enumerate}
\end{lemma}
\begin{proof}\label{proof3.1}
Let $p=4\omega +3$ and $q=4\omega^{'}+1$. Then
\begin{equation*}
\begin{aligned}
 (0,0)=&\frac{(p-2)(q-2)+1}{4}\\
 =&\frac{(4\omega +1)(4\omega^{'} -1)+1}{4}\\
 =&\frac{16\omega\omega^{'}+4\omega+4\omega^{'}}{4}\\
 =&4\omega\omega^{'}+\omega+\omega^{'}\\
 \equiv&\omega^{'}+\omega\pmod{2}
 \end{aligned}
\end{equation*}
\begin{equation*}
\begin{aligned}
(0,1)=&\frac{(p-2)(q-2)-3}{4}\\
 =&\frac{(4\omega +1)(4\omega^{'} -1)-3}{4}\\
 =&\frac{16\omega\omega^{'}-4\omega+4\omega^{'}-4}{4}\\
 =&4\omega\omega^{'}-\omega+\omega^{'}-1\\
 \equiv&\omega^{'}+\omega+1\pmod{2}
 \end{aligned}
\end{equation*}
The proof for the condition of $p=4\omega +1$ and $q=4\omega^{'}+3$ is similar.
\end{proof}

\subsection{Self-Dual Codes Over GF(2)}

\begin{theorem}\label{lemma self_dual}
Let $\frac{(p-1)(q-1)}{4}$ be even and $\frac{p+q}{4}$ be odd. Then the following holds:
\begin{enumerate}[1)]
\item  $p=4\omega +1$ and $q=4\omega^{'}+3$\\
The codes with generator matrices $GP_n(1, 0, 1, 0, 1)$ and $GP_n(1,  0,1, 1, 0)$ over \rm F(2) are self-dual of length $2n$.\\
The codes with generator matrices $GB_n(0, 0, 1, 0, 1, 0)$ and $GB_n(0, 0, 1, 0, 0, 1)$ over \rm GF(2) are self-dual of length $2n+2$.
\item  $p=4\omega +3$ and $q=4\omega^{'}+1$\\
The codes with generator matrices $GP_n(1, 1, 0, 0, 1)$ and $GP_n(1, 1,0, 1, 0)$ over \rm GF(2) are self-dual of length $2n$.\\
The codes with generator matrices $GB_n(0, 0, 0, 1, 1, 0)$ and $GB_n(0, 0, 0, 1, 0, 1)$ over \rm GF(2) are self-dual of length $2n+2$.
\end{enumerate}
\end{theorem}

\begin{table}[!hbpt]
\caption{SOME CODES OVER GF(2)} \label{t01}
\begin{center}
\begin{tabular}{|p{1.5cm}<{\centering}|p{0.15cm}<{\centering}|p{0.15cm}<{\centering}|p{2.75cm}<{\centering}|p{1.5cm}<{\centering}|}
\hline\hline
Code &p  &q &Construction& Comments\\
\hline
$[70,35,10]$  &5 &7 &$GP_{35}(1, 0, 1, 0, 1)$ &almost optimal\\
\hline
$[72,36,12]$ &5 &7  &$GB_{36}(0, 0, 1, 0, 1, 0)$ &highest know\\
\hline
\end{tabular}
\end{center}
\end{table}

\begin{proof}
If $p=4\omega +1$ and $q=4\omega^{'}+3$. Then
\begin{equation*}
\begin{aligned}
\alpha +&n= 4(4\omega\omega^{'}+\omega+3\omega^{'})+3 \equiv 1\pmod{2}, \\
-\alpha&+ m_0+(p-1)m_1+(q-1)m_2\\
+&\frac{(p-1)(q-1)}{2}(m_3 +m_4)\equiv 0\pmod{2},
\end{aligned}
\end{equation*}
\begin{equation*}
\begin{aligned}
&D_0(0,1,0,1,0)\equiv 0\pmod{2},\\
&D_1(0,1,0,1,0)\\
=&D_2(0,1,0,1,0)\equiv 1\pmod{2},\\
&D_3(0,1,0,1,0)\\
\!=&(0,0)(m_3^{2}+m_4^{2})+((1,0)+(0,1))m_3m_4\\\!
=&0\times(0+1)+(1+0)\equiv 1\pmod{2}.
\end{aligned}
\end{equation*}
By Theorem 2, the code with generator matrix $GB_n(0,0,1,0,1,0)$ over \rm GF(2) is self-dual of length $2n+2$. The proof that the code with generator matrix $GB_n(0,0,0,1,0,1)$ over GF(2) is self-dual is similar. \\
Under the same conditions, we have
\begin{equation*}
\begin{aligned}
&D_0(1,0,1,0,1) \equiv 1\pmod{2},\\
&D_1(1,0,1,0,1)=D_2(1,0,1,0,1) \equiv 0\pmod{2},\\
\!&D_3(1,0,1,0,1)=(0,0)(m_3^{2}+m_4^{2})+((1,0)+(0,1))m_3m_4\\\!
=&0\times(0+1)+0\times(1+0) \equiv 0\pmod{2}.
\end{aligned}
\end{equation*}
 Thus the code with generator matrix $GP_n(1, 0, 1, 0, 1)$ over GF(2) is self-dual of length $2n$. The rest of this theorem can be proved similarly.
\end{proof}

[72, 36, 12] codes also can be obtain by $GB_{36}(0,0,0,$ $1,0,1)$ with $p=7$ and $q=5$ over GF(2). The structures of those codes are similar. The \textbf{highest know} means that the code meets the highest know minimum distance for its parameters.

\subsection{Self-Dual Codes Over GF(4)}

In this subsection, we give four constructions of infinite families of self-dual codes over GF(4). Let $u$ be the fixed primitive element of GF(4) satisfying $u^2+u+1=0$ , then we have the following results.
\begin{theorem}\label{lemma self_dual}
Let $\frac{(p-1)(q-1)}{4}$ be even and $\frac{p+q}{4}$ be even. Then the following holds:
\begin{enumerate}[1)]
\item  $p=4\omega +1$ and $q=4\omega^{'}+3$\\
The codes with generator matrices $GP_n(1, 0, 1,u+1, u)$ and $GP_n(1,$ $ 0, 1,u, u+1)$ over \rm GF(4) are self-dual of length $2n$. \\
 The codes with generator matrices $GB_n(0, 0, 1, 0,$ $ u, u+1)$ and $GB_n(0, 0,1, 0,u+1, u)$ over \rm GF(4) are self-dual of length $2n+2$.
\item  $p=4\omega +3$ and $q=4\omega^{'}+1$\\
The codes with generator matrices $GP_n(1,1,0,u+1,u)$ and $GP_n(1,1,0,u,u+1)$ over \rm GF(4) are self-dual of length $2n$. \\
The codes with generator matrices $GB_n(0,0,0,1,u+1,u)$ and $GB_n(0,$ $0,0,1,u,u+1)$ over \rm GF(4) are self-dual of length $2n+2$.
\end{enumerate}
\end{theorem}
\begin{table}[!hbpt]
\caption{SOME CODES OVER GF(4)} \label{t01}
\begin{center}
\begin{tabular}{|p{1.23cm}<{\centering}|p{0.25cm}<{\centering}|p{0.25cm}<{\centering}|p{3.5cm}<{\centering}|p{1.5cm}<{\centering}|}
\hline\hline
Code &p &q &Construction &Comments\\
\hline
$[30,15,6]$  &3& 5 &$GP_{15}(1, 1, 0, u+1, u)$  &\\
\hline
$[32,16,8]$  &3  &5  &$GB_{16}(0, 0, 0, 1, u+1,u)$ &almost optimal \\
\hline
\end{tabular}
\end{center}
\end{table}
\begin{proof}
If $p=4\omega +1$ and $q=4\omega^{'}+3$,
\begin{equation*}
\begin{aligned}
&\alpha +N = 4(4\omega\omega^{'}+\omega+3\omega^{'})+3 \equiv 1\pmod{2}, \\
&-\alpha+ m_0 +(p-1)m_1+(q-1)m_2\\
&  \qquad +\frac{(p-1)(q-1)}{2}(m_3 +m_4)\equiv 0\pmod{2},\\
&D_0(0,1,0,u,u+1) \equiv 0\pmod{2},\\
&D_1(0,1,0,u,u+1)\\
=&D_2(0,1,0,u,u+1)\equiv 1\pmod{2},\\
&D_3(0,1,0,u,u+1)\\
=&(0,0)(m_3^{2}+m_4^{2})+((1,0)+(0,1))m_3m_4\\
=&(u+1+u)+(1+0)\times (u+1+u)\equiv 1\pmod{2}
 \end{aligned}
\end{equation*}
By Theorem 2, the code with generator matrix $GB_n(0,0,1,0,u,u+1)$ over GF(4) is self-dual of length $2n+2$.\\
 The rest of this theorem can be proved similarly.
\end{proof}

\section{Conclusions}
In this paper, we construct several infinite families of self-dual codes based on the generalized cyclotomic sets of order two. Because of several good randomness properties, some of them have good parameters. These can enrich the choices of methods to construct good self-dual codes. We believe that these constructions will lead to good self-dual codes. It seems that the construction method by generalized cyclotomic classes of higher order is a rich source to obtain codes with good parameters.

\section*{Acknowledgments}
The work is supported by Fundamental Research Funds for the central Universities(No.17CX02030A), Shandong Provincial Natural science Foundation of China (N0.ZR2016FL01, No.ZR2017MA001, No.ZR2019MF070), Qingdao application research on special independent innovation plan project (o.16-5-1-5-jch), Key Laboratory of Applied Mathematics of Fujian Province University(Putian University)(No.SX201702,No.SX201806), Open Research Fund from Shandong provincial Key Laboratory of Computer Network, (No. SDKLCN-2017-03) and The National Natural Science Foundation of China(No.2017010754).The authors are grateful to the anonymous reviewers for valuable comments.

%\addcontentsline{toc}{chapter}{\protect\numberline{}{REFERENCES}}
%\bibliography{ijns}

\begin{thebibliography}{1}

\bibitem{Alahmadi2017On}
A. Alahmadi, C. G¨¹neri, B. Ozkaya, H. Shoaib, and P. Sol¨¦.
\newblock ``On self-dual double negacirculant codes,''
\newblock {\em Discrete Applied Mathematics}, 222(C):205--212, 2017.

\bibitem{Alahmadi2018On}
A. Alahmadi, F. Ozdemir, and P. Sol¨¦.
\newblock ``On self-dual double circulant codes,''
\newblock {\em Designs Codes Cryptography}, 86(6):1257--1265, 2018.

\bibitem{Arasu2001Self}
K.~T. Arasu and T.~A. Gulliver.
\newblock ``Self-dual codes over F$_p$ and weighing matrices,''
\newblock {\em IEEE Transactions on Information Theory}, 47(5):2051--2055,
  2001.

\bibitem{Bosma1997The}
W. Bosma, J. Cannon, and C. Playoust.
\newblock ``The magma algebra system i: The user language ,''
\newblock {\em J Symbolic Comput}, 24(3-4):235--265, 1997.

\bibitem{Ding1998Autocorrelation}
C. Ding.
\newblock ``Autocorrelation values of generalized cyclotomic sequences of order
  two,''
\newblock {\em IEEE Transactions on Information Theory}, 44(4):1699--1702,
  1998.

\bibitem{Ding1996Chinese}
C. Ding, D.~Pei, and A.~Salomaa.
\newblock {\em Chinese remainder theorem: applications in computing, coding,
  cryptography}.
\newblock 1996.

\bibitem{Ding2012Cyclic}
C. Ding.
\newblock ``Cyclic codes from the two-prime sequences,''
\newblock {\em IEEE Transactions on Information Theory}, 58(6):3881--3891,
  2012.

\bibitem{Ding2013Cyclic}
C. Ding.
\newblock ``Cyclic codes from cyclotomic sequences of order four,''
\newblock {\em Finite Fields Their Applications}, 23(96):8--34, 2013.

\bibitem{Du2011Trace}
X. Du and Z. Chen.
\newblock ``Trace representation of binary generalized cyclotomic sequences with
  length,''
\newblock {\em Ieice Transactions on Fundamentals of Electronics Communications
  Computer Sciences}, 94-A(2):761--765, 2011.

\bibitem{Feng2012Cyclotomic}
T. Feng  and Q. Xiang.
\newblock ``Cyclotomic constructions of skew hadamard difference sets,''
\newblock {\em Journal of Combinatorial Theory}, 119(1):245--256, 2012.

\bibitem{Gaborit2002Quadratic}
P. Gaborit.
\newblock ``Quadratic double circulant codes over fields,''
\newblock {\em Journal of Combinatorial Theory}, 97(1):85--107, 2002.

\bibitem{Garcia2010Application}
M.~Garcia-Rodriguez, Y.~Yanez, M.~J. Garcia-Hernandez, J.~Salazar, A.~Turo, and
  J.~A. Chavez.
\newblock ``Application of golay codes to improve the dynamic range in ultrasonic
  lamb waves air-coupled systems,''
\newblock {\em NDT E International}, 43(8):677--686, 2010.

\bibitem{Hu2015Autocorrelation}
L. Hu, Q.Yue, and X. Zhu.
\newblock ``Autocorrelation value of generalized cyclotomic sequences with period
  pq,''
\newblock {\em Journal of Nanjing University of Science Technology}, 2015.

\bibitem{Macwilliams1977The}
F.~J. Macwilliams and N. J.~A. Sloane.
\newblock {\em The theory of error-correcting codes}.
\newblock 1977.

\bibitem{Nebe2006Self}
G. Nebe, E.~M. Rains, and N. J.~A Sloane.
\newblock ``Self-dual codes and invariant theory (algorithms and computation in
  mathematics),''
\newblock In {\em Math Nachrichten}, 2006.


\bibitem{Rains1998Shadow}
E.~M. Rains.
\newblock ``Shadow bounds for self-dual codes,''
\newblock {\em IEEE Transactions on Information Theory}, 44(1):134--139, 1998.

\bibitem{Stanton1958A}
R. G Stanton and D. A Sprott.
\newblock ``A family of difference sets,''
\newblock {\em Canadian Journal of Mathematics}, 10(1):73--77, 1958.

\bibitem{Tao2015Fourth}
T. Zhang  and G. Ge.
\newblock ``Fourth power residue double circulant self-dual codes,''
\newblock {\em IEEE Transactions on Information Theory}, 61(8):4243--4252,
  2015.

\bibitem{Wang2016Generalized}
Q. Wang and D. Lin.
\newblock ``Generalized cyclotomic numbers of order two and their applications,''
\newblock {\em Cryptography Communications}, 8(4):605--616, 2016.

\bibitem{Yan2009Linear}
T. Yan, X. Du, G. Xiao, and X. Huang.
\newblock ``Linear complexity of binary whiteman generalized cyclotomic sequences
  of order 2k,''
\newblock {\em Information Sciences},
  179(7):1019--1023, 2009.


\end{thebibliography}
%\bibliographystyle{IJplain}

%\newpage
\section*{Biography}
\noindent {\bf Wenpeng Gao} was born in 1994 in Shandong Province of China. He was graduated from the Department of Mathematics, China University of Petroleum, China, in 2017.He is a graduate student of China University of Petroleum.\vspace*{0.2cm}\\
\noindent {\bf Tongjiang Yan}  was born in 1973 in Shandong Province of China. He was graduated from the Department of Mathematics, Huaibei Coal-industry Teachers College, China, in 1996. He received the M.S. degree in mathematics from the Northeast Normal University, Lanzhou, China, in 1999, and the Ph.D degree in cryptography from Xidian University, Xian, China, in 2007. Now he is a associate professor of China University of Petroleum. His research interests include cryptography and algebra.\vspace*{0.2cm}\\

\end{document}